\pdfoutput=1
\newif\ifFull
\newif\ifShort
\newif\ifArxiv

\Arxivtrue 
\Fulltrue  
\Shortfalse  

\newif\ifAppendix
\Appendixfalse  

\ifFull
\documentclass{llncs}
\advance \topmargin by -\headheight
\advance \topmargin by -\headsep
\evensidemargin \oddsidemargin
\else
\documentclass{llncs}
\usepackage{times}
\Shorttrue
\fi
\usepackage{color}
\usepackage{graphicx}
\usepackage{caption}
\usepackage{subcaption}
\usepackage{url}
\usepackage{algpseudocode}
\usepackage{algorithm}
\usepackage[font=normalsize]{subfig}

\newcommand{\comment}[1]{}

\makeatletter
\def\@begintheorem#1#2{\sl \trivlist \item[\hskip \labelsep{\bf #1\ #2:}]}
\def\@opargbegintheorem#1#2#3{\sl \trivlist
      \item[\hskip \labelsep{\bf #1\ #2\ #3:}]}
\makeatother

\newcommand{\parent}{{\mathsf{parent}}}

\newcommand{\size}{{\mathsf{size}}}
\newcommand{\height}{{\mathsf{height}}}
\newcommand{\true}{{\mathsf{true}}}
\newcommand{\false}{{\mathsf{false}}}
\newcommand{\issolid}{{\mathsf{isSolid}}}

\newcommand{\tag}{{\mathsf{tag}}}
\newcommand{\next}{{\mathsf{next}}}

\newcommand{\ssize}{{\mathsf{subsize}}}
\newcommand{\totalssize}{{\mathsf{total\_subsize}}}

\newcommand{\leftS}{{\mathsf{left\_spanning}}}
\newcommand{\rightS}{{\mathsf{right\_spanning}}}
\newcommand{\indeg}{{\mathsf{indegree}}}
\newcommand{\counter}{{\mathsf{counter}}}
\newcommand{\dummy}{{\mathsf{dummy}}}

\newcommand{\child}{{\mathsf{child}}}
\newcommand{\num}{{\mathsf{child\_num}}}
\newcommand{\pchild}{{\mathsf{parent\_outdeg}}}
\newcommand{\cchild}{{\mathsf{child\_outdeg}}}
\newcommand{\psize}{{\mathsf{parent\_area}}}
\newcommand{\csize}{{\mathsf{child\_area}}}
\newcommand{\direction}{{\mathsf{direction}}}
\newcommand{\up}{{\mathsf{up}}}
\newcommand{\down}{{\mathsf{down}}}
\newcommand{\getEdge}{{\mathsf{getEdge}}}
\newcommand{\troot}{{\mathsf{root}}}
\newcommand{\axis}{{\mathsf{axis}}}
\newcommand{\width}{{\mathsf{width}}}
\newcommand{\area}{{\mathsf{area}}}
\newcommand{\out}{{\mathsf{out}}}
\newcommand{\node}{{\mathsf{node}}}
\newcommand{\unit}{{\mathsf{unit}}}
\newcommand{\prevP}{{\mathsf{prevP}}}
\newcommand{\prevQ}{{\mathsf{prevQ}}}

\newcommand{\pspqtype}{{\mathsf{parent\_spq\_type}}}
\newcommand{\cspqtype}{{\mathsf{child\_spq\_type}}}
\newcommand{\source}{{\mathsf{source}}}
\newcommand{\sink}{{\mathsf{sink}}}

\newcommand{\refpoint}{{\mathsf{refpoint}}}
\newcommand{\level}{{\mathsf{level}}}

\begin{document}

\newcommand{\acks}{%
Research supported in part by the National Science Foundation
under grants
  0830149, 
  0830403, 
  1228485, 
  and 1228639 
  and by a NetApp Faculty Fellowship. 
  We would like to thank Giuseppe Di Battista for useful discussions.
}

\ifArxiv
\title{Data-Oblivious Graph Drawing Model and Algorithms}
\else
\title{Graph Drawing in the Cloud: \\
   Privately Visualizing Relational Data \\
   using Small Working Storage%
\ifShort
\thanks{\acks}
\fi
}
\fi

\author{Michael T.~Goodrich\inst{1}
\and
Olga Ohrimenko\inst{2}
\and 
Roberto Tamassia\inst{2}
}

\institute{Dept. Computer Science, Univ. of California, Irvine.
        \email{goodrich@acm.org}
       \and
Dept. Computer Science, Brown University. \email{\{olya,rt\}@cs.brown.edu}
        }

\maketitle 

\begin{abstract}
We study graph drawing in a 
cloud-computing context where data is stored externally
and processed using a small local working storage.
We show that a number of classic graph drawing algorithms can be 
efficiently implemented in such a framework where 
the client can maintain privacy while constructing a drawing of her graph.
\end{abstract}


\section{Introduction}
\label{sec:intro}
\ifFull
Technologies developed under the
paradigm of \emph{cloud computing}
enable users to access their data seamlessly across devices of vastly
different computational power.  Moreover, these technologies support
computations on a large data set using a small device whose storage
capacity is insufficient to simultaneously hold all the data, since
the data is held in its entirety on an external server. 
\fi

In this
paper, we present techniques that allow a client to efficiently
execute various classic graph drawing algorithms, and variations of
them, in a cloud computing environment, where the storage of the
graph is outsourced to an online storage service.

We are particularly interested in allowing a client to access her data
and perform computations on them in a \emph{privacy-preserving}
way. For example, an administrator for a fast-growing company may be
revising (and visualizing) the organizational chart for the leadership
of her company, and leaking this chart to the press or a rival could
negatively impact the company.  Thus, we view the storage server as an
\emph{honest-but-curious} adversary, who correctly performs the storage
and retrieval operations requested by the client, but is nevertheless
interested in learning as much from her data as possible (indeed, some
cloud computing companies are basing their business model on this
goal).

Of course, in a cloud computing scenario, the client would encrypt the
data she outsources, decrypting it when she retrieves, it and
re-encrypting it when she stores it back (using a probabilistic cipher
that is unlikely to repeat the same cipher text for a re-encryption of
the same plaintext).  But she may also be leaking information to the
server from the pattern of her data accesses to the storage server. 
For example, accessing the memory associated with a certain department while 
preparing a new organizational chart leaks the fact that that department 
is being reorganized.
So the client should additionally aim at completely hiding her access
patterns in order to achieve privacy protection for her data.

\ifFull
\subsection{Oblivious  Algorithms and Storage}
\fi
The general techniques of \emph{oblivious RAM
  simulation} and \emph{oblivious storage} allow a client to simulate
an arbitrary algorithm in such a cloud-computing environment so as to
hide both the content and access patterns for her computation 
\ifFull
(e.g., see~\cite{a-orwca-10,bmp-rosmor-11,dmn-psor-10,go-spsor-96,%
  gm-paodor-11,gmot-orsew-11,gmot-pos-12,gmot-ppgdasors-12}).  
\fi
But
these solutions involve fairly complicated simulation techniques for
generic algorithms that increase the running time of the client's
algorithm by a polylogarithmic factor when the client has a small
amount of working storage.

Privacy-preserving algorithms in the cloud computing scenario with no
asymptotic time overhead have been developed for sorting
\cite{g-rsaso-10} and for fundamental computational geometry problems
on planar point sets, including convex hull, well-separated pair decomposition,
compressed quadtree construction, closest pairs, and all nearest
neighbors~\cite{egt-ppdo-10}. These algorithms also hide the access pattern from the server and are referred to as \emph{data-oblivious}.

In this paper, we develop simple privacy-preserving algorithms for
some classic graph drawing problems that fully obfuscate the access
pattern from the data server. Our algorithms are provably data-oblivious and
utilize small working storage.
\ifFull
\subsection{Related  Work}
\fi
There are existing
web-based systems that can perform graph drawing services for clients,
such as the Brown Graph Server~\cite{bgt-gdtsw-99} and
Grappa~\cite{bmw-grappa-97}.
These differ from the framework we are describing in this paper in two ways.
First, our model involves the client storing her
data in an outsourced data
server and accessing that data remotely, whereas the web-based graph drawing
services involve a client storing her data locally
and temporarily shipping it to the server.
Second, in the framework we are describing here, the client
performs the graph drawing algorithm herself, not the server 
(because of
privacy concerns), whereas the web-based drawing services employ their own graph
drawing algorithms to produce layouts for the client.

\ifFull
Abello and Korn~\cite{ak-mgv-02} describe a system of 
clustering and hierarchical 
representations for visualizing large graphs that do not fit in main memory.
In addition, Abello {\it et al.}~\cite{Abello:1999} discuss at a high
level, in a SIGGRAPH column, some of the challenges
of dealing with the visualization of large graphs with small working storage.
Likewise, there is a large body of work on external-memory graph algorithms, 
where one solves a problem that is too big to fit in main memory
by dividing it into blocks and accessing data in way that takes advantages of
localities of reference (e.g., see~\cite{Chiang:1995,Vitter:2001}).
The present paper can be viewed as an approach for dealing with 
the challenge of drawing large graphs with small local memory, but
in a different way than using clustering or external memory approaches.
In particular, our approach involves accessing items individually, not in
blocks, and it involves computing exact drawings, not approximate, clustered,
or hierarchical drawings.
\fi

Our approach is probably most similar to prior work on computations on data
streams (e.g., 
see~\cite{Babcock:2002,raghavan1999computing,muthukrishnan2005data}).
In this model, data is presented in single stream, which arrives in an
arbitrary order and is processed in an online, read-only fashion
using a workspace of small size.
Each time an item is considered, all the 
processing involving that item has to be
completed before considering the next item.
Henzinger {\it et al.}~\cite{raghavan1999computing} 
introduce a version of this
model that allows for a small number of passes over the data
using a small workspace, but 
their approach still assumes that data is presented in 
a read-only fashion in an
arbitrary order (although they do leave as an open problem whether allowing for
alternative orderings can reduce workspace memory size in some cases).
In addition, Feldman {\it et al.}~\cite{fmsss-odssc-10} define the MUD model for 
describing MapReduce algorithms, which also involves scans and small local memory, 
but in their model scans are only over small local
memories rather than a large set of data.

In the context of graph drawing,
Binucci {\it et al.}~\cite{bbddgppsz-10} describe a framework for drawing trees
in the streaming model, where one draws trees using a single scan of the edges,
using a framework that is similar to our approach but nonetheless has some
important differences.
Specifically,
as in the traditional data streaming model,
their approach only allows for
a single scan of the edges of a tree in an order that is not 
under the control of the algorithm.
In our case, the client can make multiple scans of her data and
specify the ordering of the scan each time.
In addition, in their model, once a node is placed it cannot be moved, 
whereas we
allow for the client to make tentative assignments of coordinates in one
scan that can be refined or changed in a future scan,
since this more naturally fits the approach of cloud computing.

\ifFull
\subsection{Our Results}
\fi
To enable data-oblivious algorithms for graph drawing problems, we
introduce \emph{compressed scanning}, an algorithmic design framework based on a series of scans.  Our method
is related to the massive,
unordered, distributed (MUD) model \cite{fmsss-odssc-10} for efficient
computation in the map-reduce framework. 
We assume that the server holds a set of $n$ data items and the client has a small private working storage of size $O(\log n)$. The data items at the server are encrypted with a semantically secure (probabilistic) cipher so that it is hard for the server to determine whether two items are equal.

An algorithm for the compressed-scanning model consists of a
sequence of rounds, where in each round the entire data set is scanned in some order specified by the client. During the scan, each item
is processed exactly once by the client: first the client downloads
the item from the server into working storage; 
next, the client performs some internal-memory computation on the item and the content of the working storage; finally the item is written out to
an output stream at the server.
When a round is completed, the output stream is
either confirmed as the algorithm's output or it is used as the input
data set for the next round.  The efficiency of such an algorithm is
measured, therefore, by the number of rounds needed and the size of
the local working storage that is required.  Ideally, the number of
rounds should be $O(1)$ and the working storage should be logarithmic
or polylogarithmic in size.
As shown in Section
\ifFull
\ref{sec:shuffle},
\fi
an algorithm
designed in the compressed scanning framework can be implemented in a
data-oblivious way by randomly shuffling the items in between scans.

Using the compressed-scanning approach, we provide efficient
data-oblivious algorithms for a number of classic graph drawing
methods \cite{dett-gd-99}, including symmetric straight-line drawings and
treemap \cite{js-tmsat-91} drawings of trees, dominance drawings of
planar acyclic digraphs~\cite{dtt-arsdp-92}, and $\Delta$-drawings of
series-parallel graphs~\cite{bcdtt-hdspd-94}.  Our methods result in
privacy-preserving graph drawing algorithms whose running times are
asymptotically optimal and better than could be achieved by applying
general-purpose privacy-preserving techniques \ifFull (e.g.,
see~\cite{a-orwca-10,bmp-rosmor-11,dmn-psor-10,go-spsor-96,%
  gm-paodor-11,gmot-orsew-11,gmot-pos-12,gmot-ppgdasors-12}).  \fi


\section{Compressed-Scanning}
\label{sec:model}
In this section, we formally define
the \emph{compressed-scanning} model
for designing client-server algorithms that can be efficiently implemented using a small working storage, $W$, at the client.
We assume that the server holds an array, $S$, of $n$ data elements.

\ifFull
\subsection{Model}
\fi
An algorithm for our model consists of a sequence of $t$ rounds.
A round involves accessing each of the elements of $S$ 
exactly once in a read-compute-write operation.
This operation consists of reading an element from the server
into private working storage, using the element
in some computation, and writing a new element
to an output stream, $O$, at the server.
When a round completes, either the output stream $O$ and/or a set of values in $W$ are confirmed as the 
output of the algorithm, or we assign $S=O$ and start the next round.

This size of the working storage, $W$, is a parameter of our model, and is intended to be small (e.g., constant or $O(\log n)$).
The name of our model is derived from the fact that each round scans the set
$S$ and computations are performed using a small, or ``compressed'', amount
of working storage.
\ifFull%
Simple examples of algorithms that fit our model include the trivial
methods for summing $n$ integers in an array or traversing a linked
list from beginning to end, which can be done with a constant-size
working storage, or any algorithm in the standard data streaming
model, which would have $W$ being equal to the working storage for
that algorithm.
\fi

\ifFull
\subsection{Privacy Protection}
\label{sec:shuffle}
\fi
Suppose we are given a compressed-scanning algorithm, $A$, which runs 
in $t$ rounds
using a working storage, $W$, and a data set, $S$, of size $n$.
We can implement $A$ in a privacy-preserving way as follows.

The first essential step in ensuring privacy is the encryption of the elements in~$S$.
From now on we assume that the input stream, $S$, is stored encrypted at the server and whenever we write
elements to the output stream, $O$, we also encrypt them.
We use semantically secure encryption~\cite{goldre04b}, which takes as input the plaintext
and a random value. Thus, if the same
element is encrypted twice, the resulting ciphertexts are different.
\ifFull
This is useful when we read an element, decrypt it, possibly modify it,
re-encrypt and write it back.
\fi
With semantically secure encryption the server
will not be able to distinguish whether two data elements are equal or whether
 the output element
of a read-compute-write operation is equal to the input element.

The next step in ensuring privacy is hiding the access pattern from
the server.  For each round, $i$, of $A$, we use a new pseudo-random
permutation~\cite{Gold01}, $\pi$, to assign a random integer, $\pi(x)$ between $1$
and $n$, to each element, $x$, in~$S$.  Value $\pi(x)$ is stored
encrypted next to item~$x$.
We then perform a random shuffle (e.g.,
using an oblivious sort~\cite{g-rsaso-10}) to move each element $x$ to location
$\pi(x)$ so that the server cannot figure out where each element was
moved to. (Recall that the elements are reencrypted each time they are
accessed.) This step takes $O(n\log n)$ time.
We now put the elements in a lookup table using the $\pi(x)$ values 
as keys, since the adversary will have no way of correlating these values to the original
locations of the elements. 

Finally, we simulate round $i$, where we use $\pi(x)$ to do the lookup for 
element~$x$.
Since each element in $S$ is 
accessed exactly once in the round, each lookup is independent
and random;
hence, it cannot be correlated with previous or subsequent lookups.
For each lookup, we
do any necessary
local computation, and then write an element to our output stream.
Even if we have nothing to output, we can always write a dummy
element, for the sake of being oblivious.

In conclusion, we simulate each round of algorithm $A$ in $O(n\log
n)$ time while fully hiding the pattern of access to the items in~$S$.
Thus, the simulation of $A$ takes time $O(t n \log
n)$ and uses working space of size proportional to that of~$A$.

\begin{definition}
  A probabilistic algorithm $A$ is \emph{data-oblivious} if given two
  inputs of the same size, $I_1$ and $I_2$, the accesses that $A$
  makes to the memory for $I_1$ and $I_2$ have the same probability
  distribution.
\label{def:obl}
\end{definition}
In other words, one cannot distinguish between $I_1$ and $I_2$ by just
looking at their access patterns.  For example, consider an algorithm
that scans the elements of a sorted array and writes to
the output stream, $O$, only distinct elements.  This algorithm is not
data-oblivious since, given inputs $(1,1,1,2)$ and $(1,2,2,2)$,
the write accesses to $O$ happen after a different number of read
accesses are made to the input stream.  A data-oblivious algorithm
would write a value to $O$ for every element it reads from the input:
a dummy element if the same element as the previous one is read, and a
real one, otherwise. One can then make a simple sorting pass over $O$
to bring real items to the front of the list.  A workspace of constant size
is used to store the last read element.
\begin{theorem} \label{thm:compressed-oblivious} %
  Let $A$ be an algorithm in the compressed-scanning model for an input
  of size $n$ that uses a working space of size $k$. Algorithm $A$ can
  be simulated by a data-oblivious algorithm if the number of rounds
  and the number of elements written to the output stream at each
  round depend only on~$n$. Also, the simulation uses a working space
  of size $O(k)$ and runs in time $O(T(n) n \log n)$, where $T(n)$ is
  the running time of~$A$.
\end{theorem}
\begin{proof} (\emph{Sketch})
  Each round is simulated by reading elements from $S$, writing elements
  to $O$, and reshuffling the next input set.  Accesses to locations
  in $S$ are made only once in a random order.  This ensures that
  accesses to $S$ in a single round are data-oblivious.  Write
  accesses to $O$ are also data-oblivious, since they happen on every
  access to $S$.  After every round, the input sequence is reshuffled
  (data-obliviously); hence, one cannot correlate accesses between
  rounds as well.  Thus, accesses to $S$ and $O$ depend only on size
  of $S$ while the number of rounds is fixed by the algorithm
  regardless of~$S$. \qed
\end{proof}
In the next section we describe graph drawing algorithms that fit the
compressed-scanning model and, hence, can be implemented in a
data-oblivious manner.
\ifFull
These algorithms guarantee that their access patterns do not reveal
the combinatorial structure of the graphs that are given as inputs
(e.g., number of outgoing or incoming edges for a particular node) and
run in a constant number of rounds using $W$ of logarithmic size.  \fi


\section{Graph Drawing Algorithms}
\label{sec:gd_algs}
Most existing graph drawing algorithms are
designed without privacy concerns in mind; hence, if they are run
in a cloud-computing environment,
they can reveal potentially sensitive information from their access patterns.
For example, a recursive binary-tree drawing
algorithm implemented in the standard way can reveal the 
depth of the tree from the access patterns used for the recursion stack,
even if all the nodes in the tree are encrypted.
In this section, we present several graph drawing algorithms modified to
fit the compressed-scanning model. 
\ifFull
In order to build a graph drawing algorithm that fits this model, we
\fi %
modify the representation of the graph so that we never access
the same location more than once in the same round. For example, consider a tree represented with a set of nodes and pointers 
from each node to its children and a parent.
Traversing the tree in this case involves accessing an internal node several times depending on its degree, which
reveals information about the tree.

\ifFull
\subsection{Euler Tours in the Compressed-Scanning Model}
\label{sec:euler_traversal}
\fi
Traversing a tree in the compressed-scanning model requires that
we access each memory location exactly once; hence, we need to 
reorganize how we normally perform data accesses, since, for example, we cannot 
access a parent again 
when coming from its left child after we have already visited it and its 
right child.
Given our small private workspace, $W$,
we cannot store previously accessed nodes.
Thus, we need a representation of a tree that allows for a 
traversal where elements are accessed only once.
For this purpose, 
we construct an Euler tour over a tree that is 
based on duplicating edges and defines a left to right traversal of a tree.
Each copy of an edge contains a pointer to a copy of the next edge in the tour so we can go to the next edge without using recursion and visiting each edge of the tour only once.

For an ordered tree, $T = (V,E)$, we store an Euler tour as a set 
of items, $C$, where $|C| = 2|E|$.
Each item represents an edge of the tour and stores information related
to the tree, e.g., $\parent$, $\child$ node names, 
and the order of the child among all its siblings.
Additionally, it stores information related to the 
actual cycle of the Euler
tour:
(a) $\tag$: a unique tag associated with this item, $0 \le \tag < 2|E|$. This is used to locate and permute
items.
(b) $\direction$: $\up$ or $\down$. This indicates 
which direction in the tree we are following.
(c) $\next$: $\tag$ of the next edge in the cycle.

We assume that $\tag=0$ for the leftmost edge of the root
of~$T$.  Suppose we shuffle the items in $C$ using a permutation, $\pi$,
over the $\tag$ field.  Then a traversal of $C$ starts with an access
to location $\pi(0)$, following access to $\pi(\pi(0).\next)$. The
items in $C$ are accessed only once and the tree layout is hidden
behind the permutation~$\pi$.  Thus, the traversal is data-oblivious
and reveals only the number of edges and nodes in the tree.

\ifFull
\subsection{Computation over Euler Tour Representations}
\fi
\label{sec:euler_ssize}
Many graph drawing algorithms collect information from a tree representation
of the graph to determine the layout. 
\ifFull Such information could be the height, width, or subtree size
of each node of the tree.  \fi
We now show how one can use an Euler tour representation of a rooted
tree to compute for each node of the tree, the size (number of nodes)
of it subtree in a data-oblivious manner.

For this computation, 
we add a new field $\ssize$ for every edge in the Euler tour $C$.
The algorithm maintains in local memory, $W$, a variable, $\totalssize$, initially
set to 0. Edges in $C$ are traversed as described in the previous section.
However, every time we now read an edge, $i$, we update $i.\ssize$ with 
the value stored at $\totalssize$ and write it back.
When we are going up, i.e., $i.\direction = \up$, 
$\totalssize$ is incremented by 1.
Once the traversal finishes, we observe that for every two items, $i$ and $i'$,
that represent a traversal of the same edge, i.e., $i.\parent = i'.\parent$,
$i.\child = i'.\child$, $i.\direction = down$ and  $i'.\direction = up $,
the value $(i'.\ssize - i.\ssize)$ is the size of the subtree rooted at $i.\child$
and the final value of $\totalssize$ is $\ssize$ of the root.
However, we need to associate nodes of the tree $T$ with these values
in the compressed-scanning model as well.
For this purpose, we obliviously sort the values in $C$ using 
the fields, $\parent$ and $\child$,
to bring items that correspond to the same edge next to each other.
We then simply scan the resulting sorted list and after reading a pair of
items, $i$ and $i'$, output a pair $(i.\child, i'.\ssize - i.\ssize)$.
\ifFull
\ifArxiv
(See Figure~\ref{fig:euler_ssize}.)
\begin{figure}[ht]
	\centering
         \includegraphics[scale=0.5]{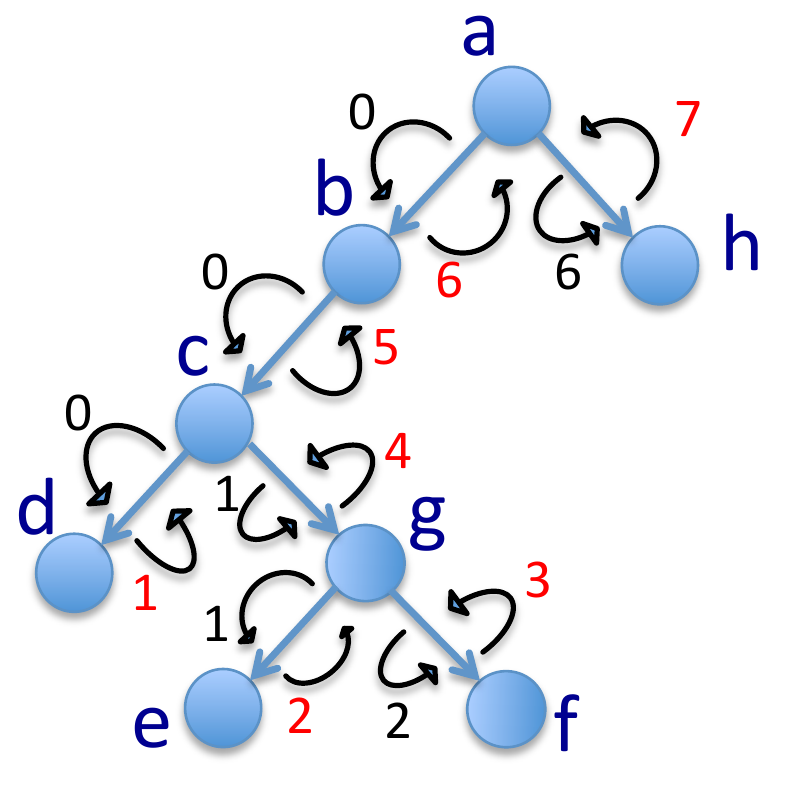}
        	\caption[]{Computing the size of the subgraph via Euler tour.
	During the tour a locally maintained variable $\totalssize$ is incremented
	when the tour goes up the tree (red numbers in the figure)
	and is assigned to currently visited edge.
	The size of the node's
	subgraph is the counter at the edge going up from this node minus the counter of the
	duplicate of this edge. For example, the size for g is (4-1) = 3.}
	\label{fig:euler_ssize}
\end{figure}
\else
\ifAppendix
(See Figure~\ref{fig:euler_ssize} in the Appendix.)
\fi
\fi
\fi

The above computation consists of two rounds: the first round reads
one item of $C$ at a time, modifies it and writes it back. The second
round starts after the sorting is complete, where items are read one at a time and
a new item is written to the output after every two reads.
We can compute the depth of each node using a similar technique.

\ifFull
\subsection{Drawing of Planar Acyclic Digraphs}
\fi
We adopt an algorithm for dominance
drawings of planar acyclic digraphs from~\cite{dtt-arsdp-92}\ifFull,
which is simple and elegant but is not data-oblivious\fi.
To find the $x$-coordinate of each node, one builds
a spanning tree based on leftmost incoming edges of the nodes
and then traverses
this tree from left to right, numbering each node in this order.
The resulting numbering of each node is its $x$-coordinate.
The algorithm to determine the $y$-coordinates uses
the rightmost spanning tree.

\ifFull \textbf{Input}: \fi We assume that the graph, $G$, is given
as a set of edges, $E$, where $e \in E$ is an edge directed from
node $a$ to $b$ storing $\indeg$, the number of incoming edges to $b$,
and $\num$, the order of $a$ among all incoming edges to $b$;
the leftmost edge has order 0.

\ifFull \textbf{Data-oblivious algorithm}: \fi Following the original algorithm,
we show how one can construct a spanning tree and number the nodes
to get the final drawing.
Our first task is to augment each edge with
information about a spanning tree of $G$. We augment $e$ with additional 
fields, $\leftS$ and $\rightS$, which are set to $\true$ or $\false$ depending
on which spanning tree $e$ belongs to. In the compressed-scanning model,
one simply accesses $e$, sets $e.\leftS$ to $\true$ if $e.\indeg$ equals $e.\num$
or $e.\rightS$ to $\true$ if $e.\num$ is 1, and writes $e$ back. 

Given annotated edges, we construct an Euler tour over each spanning tree.
Note that given that the number of nodes in $G$ is revealed, we do not need
to hide the number of edges in either of the spanning trees.
For ease of explanation, we say that we traverse an edge $\down$ when
we follow an edge of the spanning tree in its direction in $G$. The left
spanning tree is traversed starting with the leftmost outgoing edge of the root,
and rightmost outgoing edge for the right tree.
We are now ready to make a tour traversal and assign coordinates to the nodes.
We adopt a compressed version of the algorithm that minimizes the area
of the drawing and start with traversal of the left tree.
In private memory, a $\counter$ for $x$-coordinates is maintained, 
set to~$0$.
Initially, we output $(\source, 0, x)$.
For every edge $e$ that has $\direction = \down$, and
$e.\indeg >1$ or $e$ is the first traversed edge of $a$, we output $(e.b, \counter, x)$.
If $e$ has $\down$ direction but is not the first edge 
of $a$ traversed (in Euler
tour this corresponds to remembering the latest visited edge)
or is the only incoming edge to $b$, then we increment the $\counter$ by 1
and output $(e.b, \counter, x)$. If  $e.\direction$ is set to
$\up$, then we output $(\dummy, 0, x)$. The algorithm for computing 
$y$-coordinates is similar and outputs values with $(e.\parent, \counter, y)$.
Note that access pattern of reads and writes is always the same: read
an edge of the Euler tour and output a tuple of three values.

The output of the above procedure contains tuples of real and dummy values.
We can remove dummy values and bring $x$, $y$ coordinates of each node
together by obliviously sorting tuples by the first field (node name)
such that string $\dummy$ is always greater than any real node name.
The resulting list contains all dummy tuples at the end. Also, each node
has its $x$- and $y$-coordinates adjacent. 
See Figure~\ref{fig:stdom} for an example.


\begin{figure}[hbt]
  \ifShort
    \vspace{-15pt}
  \fi	
  \centering
  \begin{subfigure}[b]{0.3\textwidth}
    \centering
    \includegraphics[scale=0.33]{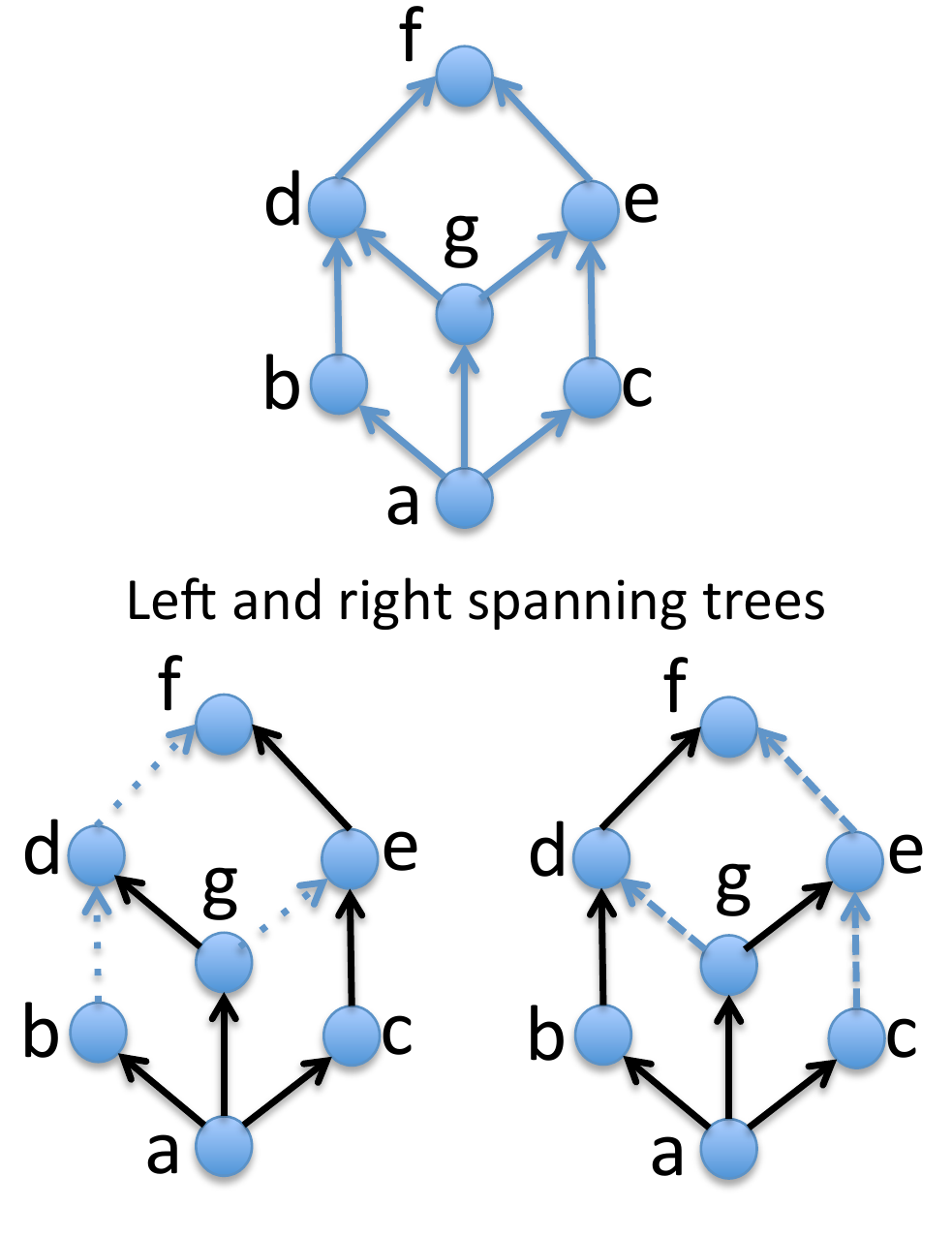}
    \label{fig:stdom1}
    \ifShort
    \vspace{-5pt}
    \fi	
    \caption{}
  \end{subfigure}
  \begin{subfigure}[b]{0.33\textwidth}
    \centering
    \includegraphics[scale=0.35]{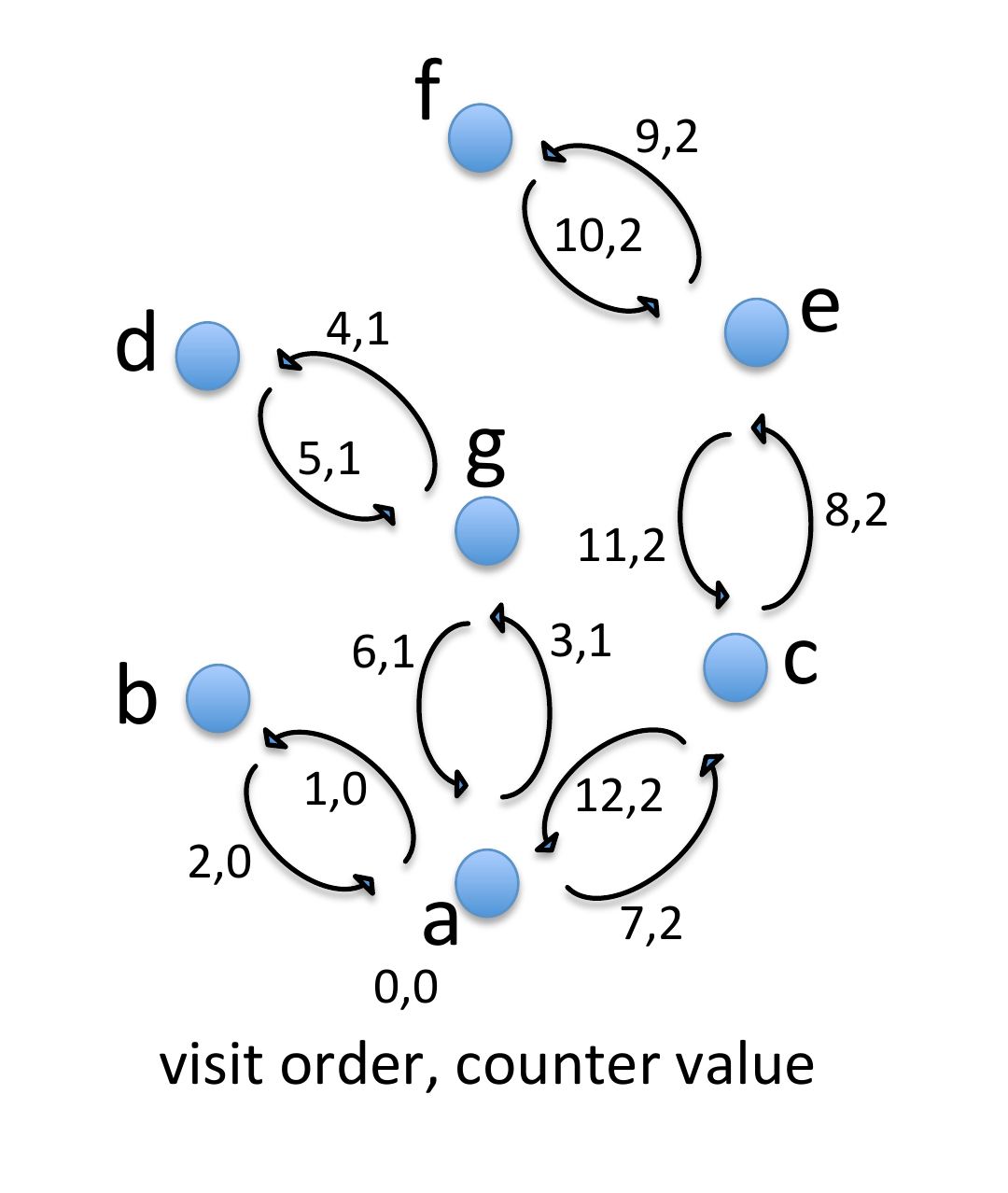}
    \label{fig:stdom2}
    \ifShort
    \vspace{-5pt}
    \fi	
    \caption{}
  \end{subfigure}
  \begin{subfigure}[b]{0.33\textwidth}
    \centering
    \includegraphics[scale=0.45]{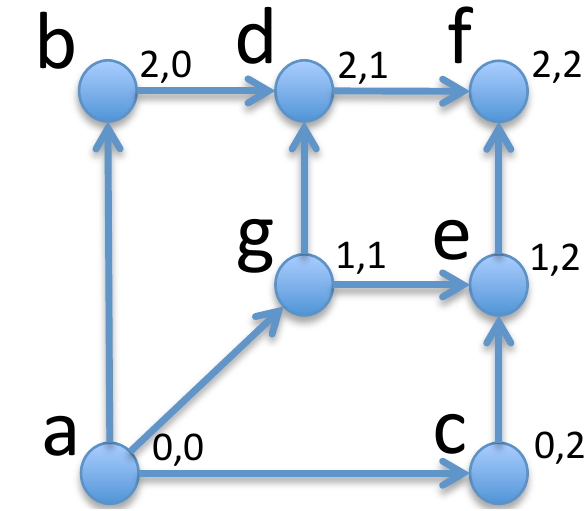}
    \label{fig:stdom3}
    \ifShort
    \vspace{-5pt}
    \fi	
    \caption{}
  \end{subfigure}	
  \ifShort
    \vspace{-5pt}
  \fi			
  \caption[]{(a) A planar acyclic digraph with its left and right spanning trees.
    (b) The order of the visit to each edge of Euler tour of the left spanning tree and the
    counter of $x$ coordinate for child nodes, e.g.,  edge $a$-$g$ is visited third and $g$ is
    assigned $x$ coordinate of 1. (c) The final drawing.}
   \label{fig:stdom}
  \ifShort
    \vspace{-21pt}
  \fi
\end{figure}


\ifFull
\subsection{Treemap Drawings}
\label{sec:treemap}
\fi
Treemaps are a representation designed for human visualization of
complex tree structures, where arbitrary trees are shown with a 2-d space-filling area.
Here, we present how one can draw a treemap using an algorithm from~\cite{js-tmsat-91}
adapted to
the compressed-scanning model. The original algorithm takes a rectangle area
and splits it vertically into two sections.
The area of the first section is enough to fit the first child, $\child_1$, of the root and
the rest is enough to fit the rest of its children. The next step is to divide the first
section among children of $\child_1$ but this time splitting the area horizontally.
The algorithm continues in the same manner for all decedents of $\child_1$.
Once finished, it proceeds to splitting the second section between second child of the root,
$\child_2$, and the rest of root's children.

\ifFull
\textbf{Input:}
\fi
A tree, $T$, where each node also contains a 
value $\area$ and the
size of a rectangle area, $w\times h$, where $T$ should be drawn.
We build an Euler tour, $C$, from $T$ and add two fields $\psize$ and  $\csize$ to each edge in $C$.

\ifFull
\textbf{Output:}
\fi
Each node is labeled with $(x,y)$ coordinates of the top-left corner, $P$,
and bottom-right  corner, $Q$, of the rectangle area where the node should be placed in.

\ifFull
\textbf{Data-oblivious algorithm:}
\fi
The original algorithm labels the nodes with values $P$ and $Q$
via pre-order traversal of $T$. The algorithm we propose here first goes down the leftmost
subtree computing values $P$, $Q$ and labeling the nodes on the way.
In private memory, it maintains only one copy of the last two assigned
values of $P$ and $Q$, $\prevP$ and $\prevQ$.
It then goes up the tree ``undoing'' all the computations made to $\prevP$ and $\prevQ$.
We do it in such a way that when going up
and reaching some node, we recover its $P$ and $Q$ values as they
were before we visited any of its children or other nodes in its subgraph.
This algorithm fits the traversal of Euler tour $C$ of the tree $T$.
When going down the tree, we read each item $i$ of tour $C$ and output $P,Q$
values corresponding to $i.\child$. However, when going up we cannot retrieve earlier
written $P,Q$ values, since this will not be data-oblivious and we reveal that we are going up,
which consequently reveals the depth of the tree. This is where ``undoing'' computations
when going up on $\prevP$ and $\prevQ$ helps. This is possible since the information
used to compute $P$ and $Q$ is stored twice in $C$: once for edge with $\direction$
set to $\down$ and once for $\up$.
\ifArxiv
The pseudocode of the algorithm appears in Algorithm~\ref{alg:treemap}.
\begin{algorithm}[h]
\begin{algorithmic}
\State $\out.\node \leftarrow \troot$, $\out.P \leftarrow [0,0]$, $\out.Q \leftarrow [w,h]$
\State \texttt{write} $\out$
\State \texttt{read} $\pi(0)$ into $e$
\Comment{Get an edge corresponding to the leftmost edge from the root of $T$}
\State $\axis \leftarrow 0$, $\unit \leftarrow w / e.\psize$
\State \Comment{$\prevP$, $\prevQ$, $\unit$, $\axis$ are maintained in private memory, $W$}
\State $\prevP \leftarrow [0,0]$, $\prevQ \leftarrow [w,h]$
\While {$e.\parent \neq \troot$ and $e.\direction \neq \up$}
   \If {$e.\direction = \down$}
     \State $\prevQ[\axis] \leftarrow \prevP[\axis] + \unit \times e.\csize$
     \State $\out.\node \leftarrow \child$, $\out.Q \leftarrow \prevQ$, $\out.P \leftarrow \prevP$
      \If {$e.\cchild = 0$ and $e.\num < e.\pchild}$
         \State $\prevP[\axis] \leftarrow \prevQ[\axis]$
         \Comment {Move the top left corner for the next child}
      \ElsIf {$e.\cchild > 0$}
         \Comment Go further down the branch
         \State $\unit \leftarrow ({\prevQ[1-\axis] - \prevP[1-\axis]})/{e.\csize}$
         \State $\axis \leftarrow 1 - \axis$
     \EndIf
   \Else
       \If {$e.\num = e.\pchild$} \Comment{Going up again. Undo previous $P$, $Q$ changes.}
          \State $\mathsf{branch\_size} \leftarrow \unit \times e.\psize$      
          \State $\unit \leftarrow ({\prevQ[1-\axis] - \prevP[1-\axis]})/{e.\psize}$
          \State $\prevP[\axis] \leftarrow \prevQ[\axis] - \mathsf{branch\_size}$
          \State $\prevP[1-\axis] \leftarrow \prevQ[1-\axis]$
          \State $\axis \leftarrow 1 - \axis$
      \EndIf
      \State $\out.\node \leftarrow \dummy$, $\out.Q \leftarrow [0,0]$, $\out.P \leftarrow [0,0]$
   \EndIf
     \State \texttt{write} $\out$
   \State \texttt{read} $\pi(e.\tag)$ into $e$
\EndWhile
\State Sort all output values by $\node$ field such that $\dummy$ values are in the end.
\end{algorithmic}
\caption{Data-oblivious algorithm to compute a treemap drawing of an arbitrary tree.}    
\label{alg:treemap}
\end{algorithm}
\fi
%
Figure~\ref{fig:treemap} shows an execution of the algorithm on a small tree.

\begin{figure}[ht]
\ifShort
\vspace{-15pt}
\fi	
\centering
\begin{subfigure}[b]{0.5\textwidth}
  \centering
  \ifShort
  \vspace{5pt}
   \fi	
   \begin{subfigure}[b]{0.5\textwidth}
       \includegraphics[scale=0.25]{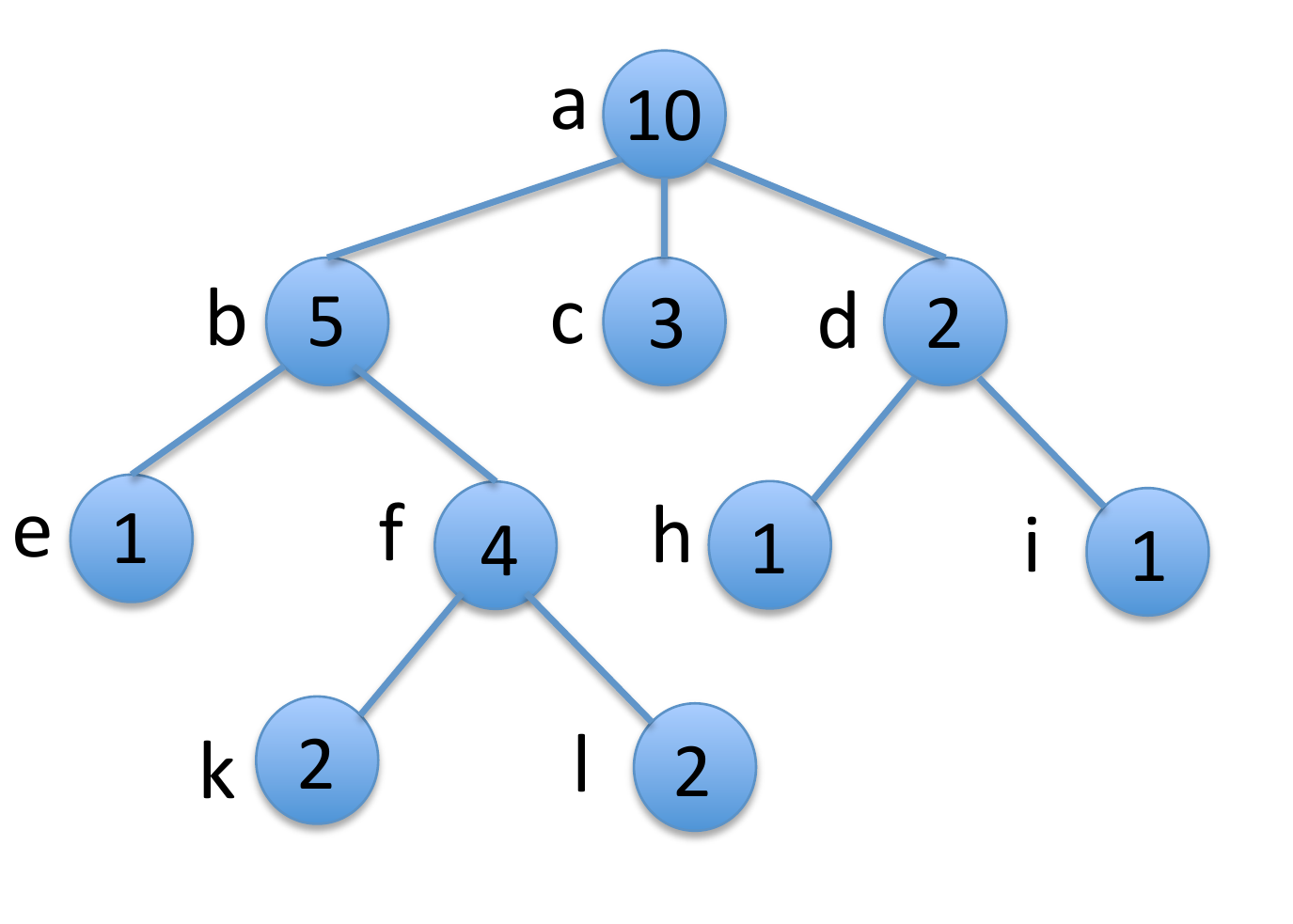}
       \label{fig:treemap1}
       \ifShort
       \vspace{-15pt}
       \fi	
       \caption{}
    \end{subfigure}
    \begin{subfigure}[b]{0.5\textwidth}
       \includegraphics[scale=0.33]{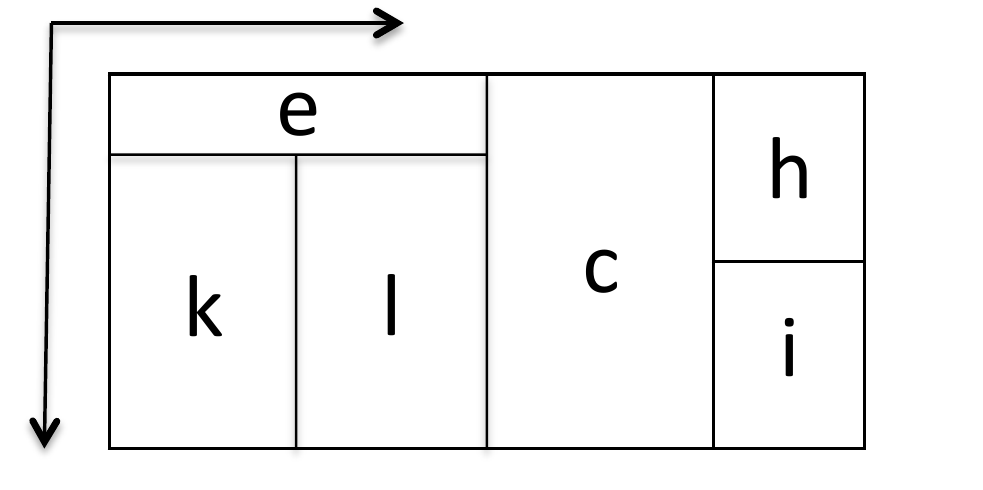}
       \label{fig:treemap2}
       \ifShort
       \vspace{-15pt}
       \fi
       \caption{}
    \end{subfigure}
\end{subfigure}%
\begin{subfigure}[b]{0.5\textwidth}
  \includegraphics[scale=0.4]{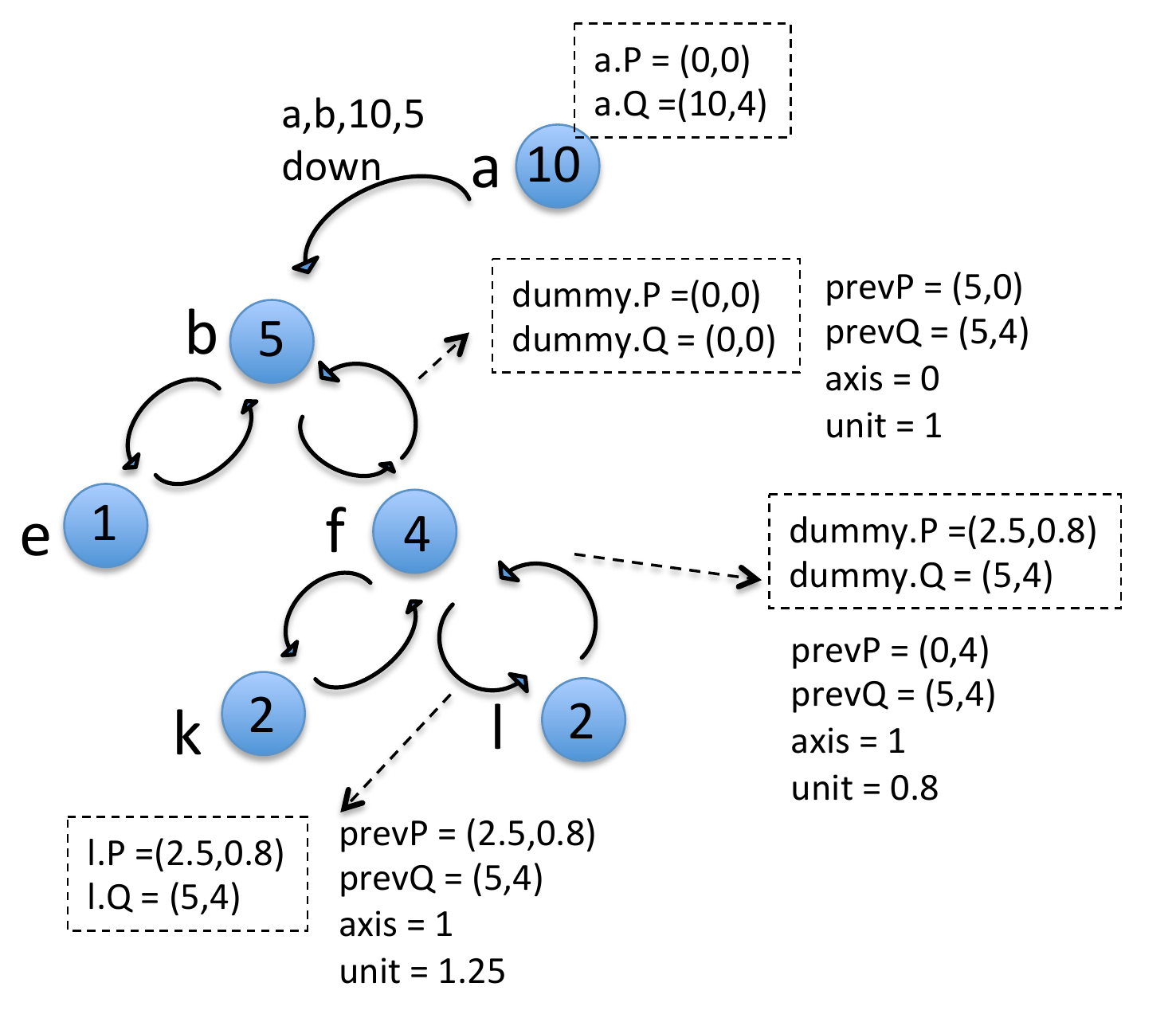}
  \label{fig:treemap3}
  \ifShort
  \vspace{-15pt}
  \fi	
  \caption{}
\end{subfigure}
\ifShort
\vspace{-15pt}
\fi	
\caption[]{Treemap graph drawing. (a) The original graph. (b) The final drawing.
\ifFull
(c) Execution of algorithm in Section~\ref{sec:treemap} on the graph in (a) on a 10$\times$4 rectangle area.
\fi
The values in dashed rectangles are written for every edge and are never accessed. Variables $\prevP$,
$\prevQ$, $\axis$ and $\unit$ are kept in memory.}
\label{fig:treemap}
\ifShort
\vspace{-18pt}
\fi
\end{figure}


\ifFull
\subsection{Series-Parallel Graphs}
\fi
A series-parallel (SP) graph is a directed acyclic graph that can be decomposed recursively
into a combination of series-parallel digraphs.
The base case of such a graph is a simple directed edge. A series composition consists of two
series-parallel graphs $G_1$ and $G_2$ where
the sink of $G_1$ is identified with the source of $G_2$. A parallel composition of two
series-parallel graphs $G_1$ and $G_2$ is the digraph where source of $G_1$ is identified
with the source of $G_2$ and similar for their sink nodes.
For example, consider the series-parallel digraph shown in Figure~\ref{fig:spq1}.
\ifFull
The subgraph $S'$ induced by its edges
$c$-$d$ and $d$-$a$ is a series composition of graphs $c$-$d$ and $d$-$a$.
While $S'$ and edge $c$-$a$ is a parallel composition.
\fi

An SP graph $G$ can be represented with a binary tree (SPQ tree) with
three types of nodes, $S$, $P$ and $Q$.
$Q$ nodes are leaves of the tree
and correspond to individual edges of $G$. An internal node
is of type $P$ if it is a parallel composition of the children digraphs.
If a node corresponds to a series composition it is called $S$ node.
Here, we use a right-pushed embedding of $G$ such that a transitive edge
in parallel composition is always embedded on the right. (Figure~\ref{fig:spq2}
shows the SPQ tree of the graph of Figure~\ref{fig:spq1}.) 

\ifFull
\textbf{Original  $\Delta$-drawing algorithm:}
We adopt the $\Delta$-drawing algorithm from~\cite{bcdtt-hdspd-94}.
This algorithm recursively produces a drawing of $G$
inside a bounding triangle $\triangle(G)$ which is isosceles and right-angled.
In the drawing of a series composition, the two bounding triangles,
$\triangle(G_1)$ and $\triangle(G_2)$,
are placed one on top of another and, hence, produce a bounding triangle big enough to
fit them both. For a parallel composition, $\triangle(G_2)$ is placed on the right
of  $\triangle(G_1)$ and a larger triangle is drawn to fit this parallel composition.
The algorithm works by traversing the SPQ tree and identifying the size of the bounding
triangles of each node. The length of the hypotenuse, $b$, is enough to store
this information. Each $Q$ node is assigned a triangle with $b=2$, while for
series and parallel nodes $b$ is the sum of $b$ values at the children nodes.
When traversing the tree we also compute value $b'$, which makes sure that in a
drawing of a parallel graph $G$ the edge that goes from the source of $G$ to $G_1$,
the left subgraph of the composition,
does not intersect the drawing of $G_2$.
This value $b'$ for a $Q$ node is simply $b$,
for $S$ node it is $b'(\triangle(G_1))$ and for $P$ node it is $b'(\triangle(G_1)) + b'(\triangle(G_2))$. Note that
for a parallel node it is the sum of $b'$ values of both graphs since we want to make sure
that if subgraph $G$ is later a part of a parallel composition no node will intersect either
$G_1$ or $G_2$. If $G$ is a transitive edge then $b'(\triangle(G)) = b(\triangle(G))$.
(See Figure~\ref{fig:spq2}.)

Once $b$ and $b'$ are computed for every node, i.e., every bounding triangle,
the algorithm computes the $(x,y)$ value of the bottom node of each triangle.
The outer most triangle is positioned at $(0,0)$.
Given coordinates $(x,y)$ of a triangle corresponding to the $S$ node
with hypotenuse of size $b$ and children with hypotenuses
$b_1$ and $b_2$, we place the first triangle
at $(x,y)$ and second at $(x,y+b_1)$.
Given coordinates $(x,y)$ of a triangle corresponding to a parallel node,
we place the first triangle
at $(x-0.5b_2,y+0.5b_2)$ and second at $(x,y+b'_1)$.
Given that we know the coordinates of each triangle, we can now assign coordinates
for individual nodes. The source of $G$ is placed at (0,0) and sink is placed at $(0, b(\triangle(G)))$.
We then look at each $\node$ in $G$ and place it at $(x, y+b(\triangle(G_\node)))$
where $G_\node$ is a subgraph and $\node$ is its sink.  (See Figure~\ref{fig:spq3} for an example.)


We are now ready to explain the algorithm in compressed-scanning model.
\fi

\ifFull
\textbf{Input:}
\fi
SPQ tree from a right-pushed embedding of SP digraph $G$ and
nodes that are annotated as $S$, $P$ or $Q$. 
We convert this tree into an Euler tour
with addition of parent and child node type:
$\pspqtype$ and $\cspqtype$ which are either $S$, $P$ or $Q$.

\ifFull
\textbf{Data-oblivious algorithm}:
\fi
\ifShort
We adopt the $\Delta$-drawing algorithm from~\cite{bcdtt-hdspd-94}.
The algorithm
makes several computations over the tree to annotate the nodes of 
the SPQ tree
with values $b$, $b'$ and $(x,y)$.
Here $b$ is the length of the hypotenuse of the
bounding triangle of a node and $b'$ stores a distance
between parallel drawings to make sure they do not intersect.
\else
The above algorithm
makes several computations over the tree to annotate the nodes of 
the SPQ tree
with values $b$, $b'$ and $(x,y)$.
\fi
Value $b$ can easily be computed in the same manner as we computed
the subgraph size
\ifFull
in Section~\ref{sec:euler_ssize}.
\fi
Value $b'$ of the left child is added
only for parents of $P$ nodes. When an Euler tour is going up the tree we can always check
the value of $\pspqtype$ to know if $b'$ of the left subgraph should be carried to the right one.
Coordinates $(x,y)$ for each node are computed from a small modification
of the Euler tour: the left child needs know value $b(\triangle(G_2))$ and
right child needs to know $b'(\triangle(G_1))$. It is easy to do this by always reading the next edge and remembering
the last edge.

Given that we know the coordinates of each triangle, we can now assign coordinates
for individual nodes. Recall that every leaf node of SPQ tree is associated with an edge
while an internal node is either a DAG or a path of edges in the subtree rooted at this node.
Hence, we can associate each internal node of SPQ tree, and edges in the corresponding
Euler tour, with two nodes of the series-parallel graph
that correspond to the source and the sink of the underlying subgraphs.
Given a $\parent$ node of SPQ tree and
source and sink nodes of its children, $c_1$ and $c_2$, if $c_1^\sink$ and $c_2^\source$ are equal
then node $c_1^\sink$ is placed at $(c_1.x, c_1.y+c_1.b)$. Otherwise, we output a dummy.


%
\begin{figure}[t]
\centering
\begin{subfigure}[b]{0.15\textwidth}
  \centering
  \includegraphics[scale=0.27]{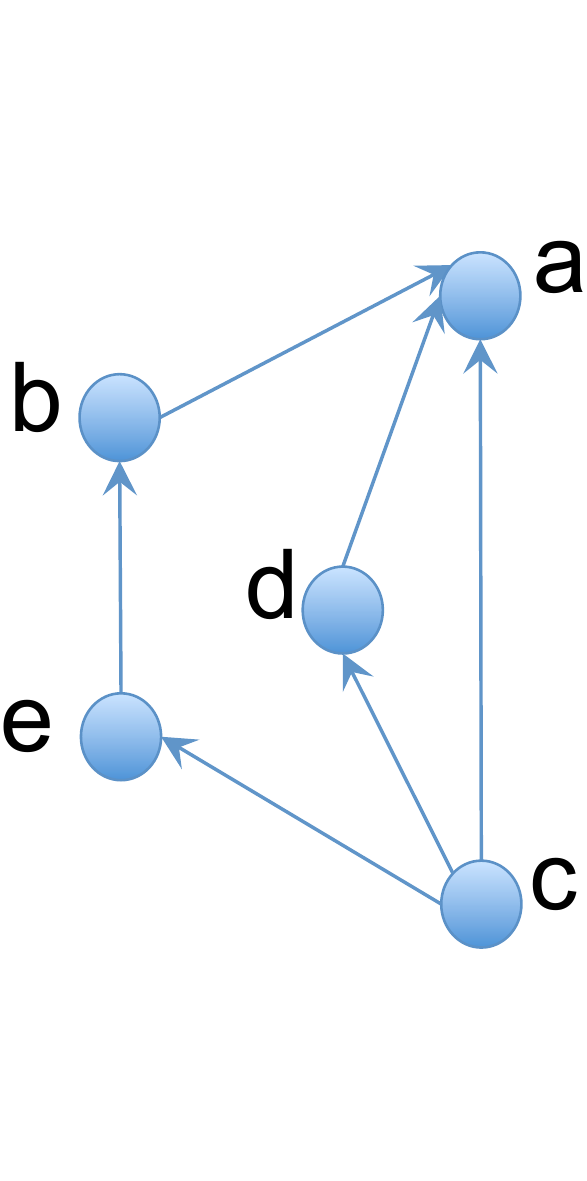}
  \caption{}
  \label{fig:spq1}               
\end{subfigure}
\begin{subfigure}[b]{0.41\textwidth}
  \centering
  \includegraphics[scale=0.41]{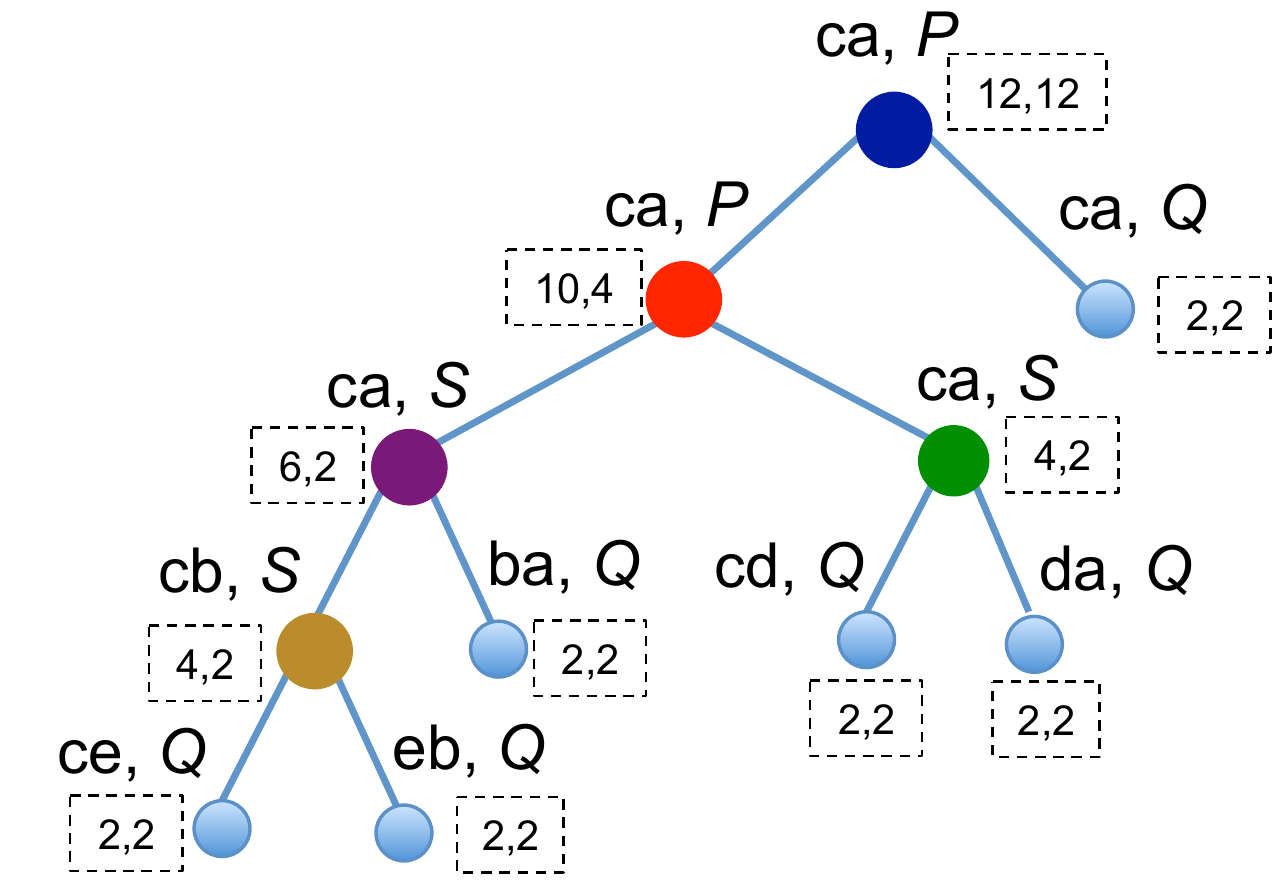}
  \caption{}
  \label{fig:spq2}
\end{subfigure}
\begin{subfigure}[b]{0.4\textwidth}
  \centering
 \includegraphics[scale=0.27]{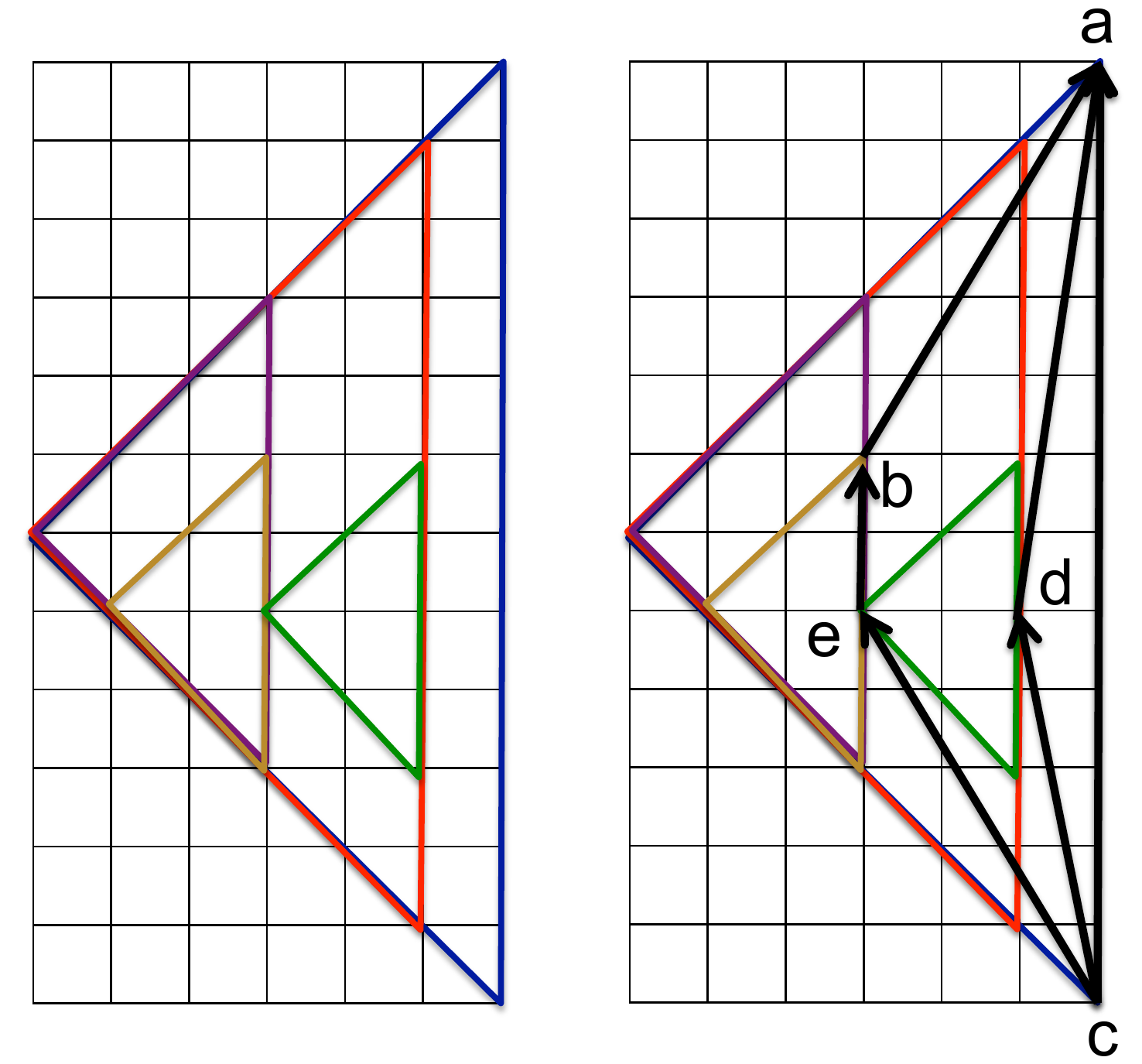}
  \caption{}
  \label{fig:spq3}
\end{subfigure}
\hspace{-15pt}
\ifShort
  \vspace{-5pt}
\fi		
\caption[]{(a) A series-parallel graph. (b) SPQ tree representation annotated with values $b$ and $b'$
(dashed rectangles).
(c) The final drawing.}
\label{fig:spq}
\ifShort
  \vspace{-20pt}
\fi
\end{figure}

\ifFull
\subsection{Drawing Trees with Bounding Rectangles}
\fi
\label{sec:bound_rect}
In this section, we present an algorithm that draws a binary tree $T$
using a bounding rectangle approach from~\cite{cdtt-dgdts-95}, adapted
to the compress-scanning model. 
This algorithm is
slightly different from the approaches we took in previous sections
and involves a more complex way of converting it to fit
data-oblivious mode. The original algorithm recursively assigns bounding
rectangles to nodes of the tree. A leaf node is assigned a rectangle of size 2$\times$1,
while an internal node is assigned a rectangle that fits the bounding rectangles
of its children. Each rectangle is represented by its $\width$, $\height$
and $(x,y)$ coordinate of the left top corner, which we refer to as reference point $\refpoint$.
For leaves, the $\width$ is $2$, and $\width$ of internal nodes
is the sum of the $\width$ of its children.
The $\height$ of the rectangle is defined as $1+\max_i \child_i.\height$.
The bounding rectangle of the
root node is assigned to $\refpoint$ of $(0, \mathsf{tree\_height})$.
The $\refpoint$ of $i$th child of node $p$ is assigned
to $(p.\refpoint.x + \sum_{j < i} \child_j.\width, p.\height)$.
Each leaf node $l$ is then assigned to a coordinate
$(x,y) = (l.\refpoint.x + l.\width/2,  l.\refpoint.y)$.
An internal node is placed between its children, hence,
a node $l$ with children $\child_i$ $(i = 1,2)$ is assigned to a coordinate
$(\sum \child_i.x/i, l.\refpoint.y)$.

\ifFull
\textbf{Data-oblivious algorithm:}
\fi
Here, an Euler tour over $T$ is only sufficient to compute $\width$,
$\level$ and $\refpoint$ values. Computing $(x,y)$ coordinate of internal nodes
involves knowing the left and the right coordinate of its children which can only be computed
when subgraphs of both children are processed. If we use Euler tour
traversal we need to store the coordinates computed in the left subtree while processing
the right subtree. Given that we only allow for small private workspace, we cannot store them
internally, since we may need to store coordinates for several levels of the tree
which in the worst case can be linear in the number of nodes in the tree.
Indeed, in our previously described methods, we only store a constant number
of values when traversing a tour.
Therefore, in this section we propose a different technique
that is based on a dashed-solid representation. This representation allows us
to store only $\log(n)$ coordinates in the worst case, which
fits our compressed-scanning model.
This algorithm is an example that one can carry out more involved computations
on graphs with sublinear private space.

The \emph{dashed-solid} representation of the tree splits edges into dashed and solid.
An edge $\parent$-$\child_i$ is solid if $\parent.\ssize/2 < \child_i.\ssize$.
Otherwise, an edge is dashed.
If $\ssize$ of the children is the same, the right edge is solid and the left one is dashed.
The invariant of this representation is that a parent node has a solid
edge to only one of its children and the corresponding subtree is equal or larger than
the subtree of the sibling (if one exists). The main property of the dashed-solid assignment is that
the length of the longest dashed path is $O(\log n)$.
Note that given that we can compute $\ssize$ easily
\ifFull
(Section \ref{sec:euler_ssize}),
\else
\fi
assigning edges to dashed or solid is trivial using another Euler tour traversal.

Given a dashed-solid representation, we compute $(x,y)$ coordinates
by creating a tour around the tree where edges are accessed in a specific order.
The traversal first goes down only via solid edges,
when a leaf is reached we traverse up until a node with a dashed edge is reached.
We then follow this subtree also traversing its solid edges first.
To construct this traversal one needs to store with every node which of its children is
solid.
The coordinates $(x,y)$ are computed as follows.
We follow a solid edge path until a leaf $l$ is reached and then the leaf node
is assigned to coordinate
$(l.\refpoint.x + l.\width/2, 0)$. We remember this coordinate as $s$
in the private memory.
When going up, if the parent node $p$ does not have any other children,
then we assign it
to $(s.x, s.y)$ and continue traversing up the tree.
If the node $p$ indeed has a dashed edge, then we traverse the subtree of the dashed edge.
Once this traversal is finished, $p$ is assigned to $((s.x + d.x)/2, 1+\max (s.y,d.y))$ where
$d$ is the coordinate of the child from the dashed edge. We now discard
old value of $s$ and $d$ and assign $s$ to just computed coordinate of $p$
and keep going up. The same process is applied when we go down the subtree
from the dashed edge, except we need to remember $s$ coordinate
for every subtree rooted at a solid edge we encounter going down.
Since the longest path starting from a  dashed edge is
bounded by $O(\log n)$ we can easily fit a stack of corresponding $s$ coordinates
into our small workspace~$W$.
\ifAppendix
\else
\ifArxiv
(See Figure~\ref{fig:bound_rect} for an example of the drawing.)
\begin{figure}[ht]
	\centering
         \includegraphics[scale=0.5]{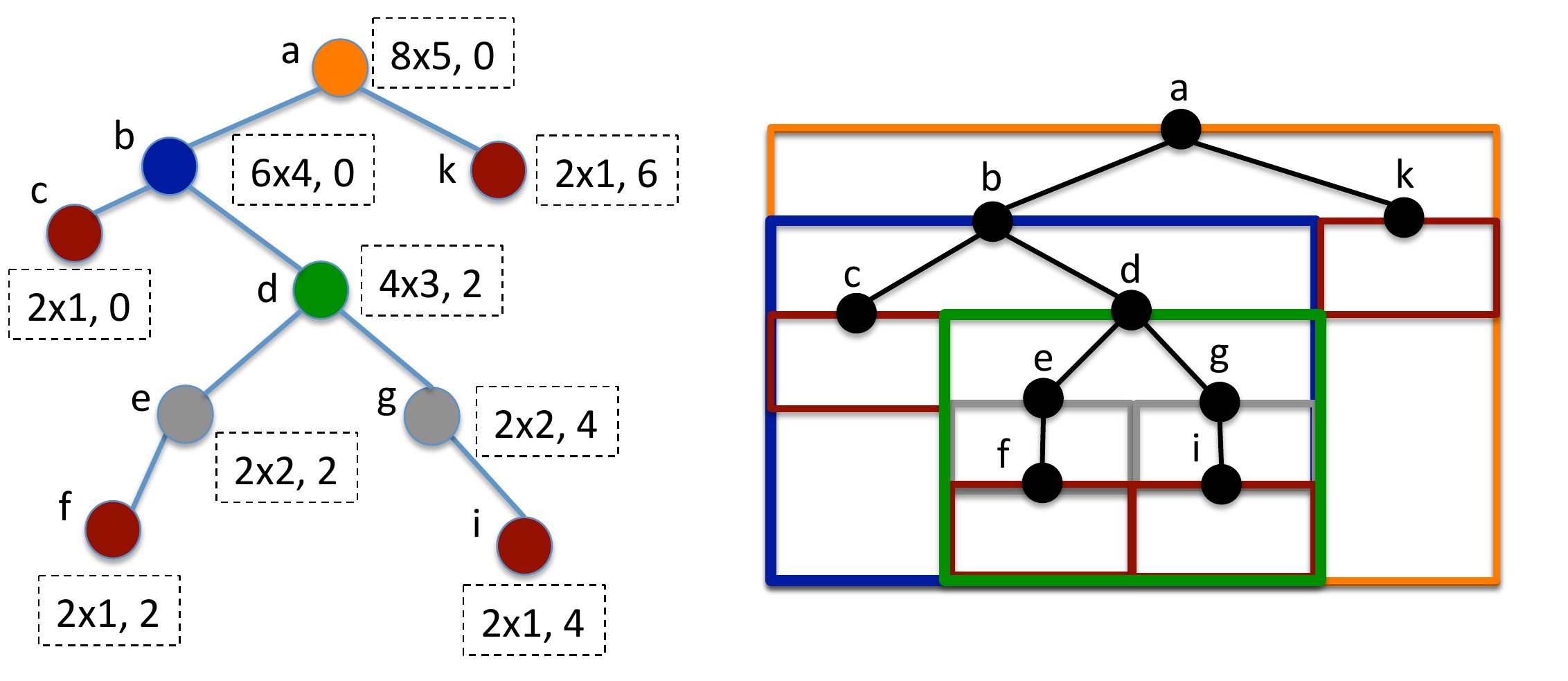}
        	\caption[]{Example of the algorithm in Section~\ref{sec:bound_rect}. Tree annotated with the
	width, height and $x$ coordinate of the top left corner of the bounding rectangle ($\refpoint$). Figure on the right
	is the resulting drawing.}
	\label{fig:bound_rect}
\end{figure}
\fi
\fi
We note that this
algorithm can be extended to n-ary trees if we store dashed edges of a single
node in a balanced binary tree (see~\cite{cdtt-dgdts-95} for details).
\ifFull
\subsection{Summary} 
\fi
The following theorem summarizes the results of this section.
\begin{theorem}
  The drawing algorithms described in this section are data-oblivious
  according to Definition~\ref{def:obl} and run in time $O(n \log n)$.
  Also, the private working space has size $O(\log n)$ for the
  bounding-rectangle tree-drawing algorithm (where $n$ is the size of
  the tree) and has size $O(1)$ for the other algorithms.
\end{theorem}
 
Since we have given algorithms in the compressed-scanning model, the
theorem follows from Theorem~\ref{thm:compressed-oblivious}. All of
the algorithms perform a constant number of Euler tours. In the beginning of
Section~\ref{sec:gd_algs} we showed that an Euler tour can be
implemented with a single-round compressed scan, where, from the
server's perspective, the items associated with the edges of the tour
are accessed in random order and only once.

Our algorithms hide the combinatorial structure and layout of the
graphs, while the number of edges and vertices is revealed. One can
achieve even stronger privacy if dummy edges and nodes
are added. From the point of view of the model, the input $S$ is a
larger set of elements and the running time of algorithm $A$ increases
as well.


\section{Conclusions and Open Problems}
\label{sec:conclusions}
We introduce the compressed-scanning technique for designing
data-oblivious algorithms in a cloud-computing
environment.  
\ifFull
In a nutshell, this technique involves specifying an
algorithm as a series of scans where data is processed using a small
working storage.  Using this technique, we show how to implement
classic drawing algorithms for trees, series-parallel graphs, and
planar $st$-digraphs (and variations of these algorithms) so that the
client needs only a small amount of working storage (constant or
logarithmic in the size of the data set) and can fully protect the
privacy of the graph and of its layout, beyond what can be
accomplished by encryption alone.
\fi
Open problems include finding other applications of this technique
and developing alternative data-oblivious approaches for graph drawing. For
example, it is not known how to compute in a data-oblivious way
$st$ orientations and $st$-numberings, used for visibility representations
of planar graphs~\cite{tt-uavrp-86}, 
or canonical orderings \cite{fpp-hdpgg-90}, used for planar straight-line drawings.


\ifFull
\section*{Acknowledgments}
  \acks
\fi

\bibliographystyle{abbrv}
\bibliography{refs,crypto,geom,goodrich}

\ifFull
\ifArxiv
\else
\clearpage
\input{appendix}
\fi
\fi

\end{document}